\documentclass[11pt,letter]{article}

\usepackage{fullpage}
\usepackage[utf8x]{inputenc}
\usepackage{amsmath, amsthm}
\usepackage{wrapfig,graphicx,amssymb,textcomp,array,amsmath,enumitem}
\usepackage{yhmath}

\usepackage{color}
\usepackage{changepage}
\usepackage{algpseudocode} 
\algtext*{EndWhile}% Remove "end while" text
\algtext*{EndIf}% Remove "end if" text
\algtext*{EndFor}% Remove "end for" text
\usepackage{algorithm}
\title{Compatible 4-Holes in Point Sets}

\author{Ahmad Biniaz\thanks{This research is supported by NSERC.}\and
Anil Maheshwari\and Michiel Smid}

\renewcommand\footnotemark{}

\newcommand{\removed}[1]{}
\newcommand{\etal}{{\em et al.}}

\newcommand{\cone}[3]{C(#1{:}#2,#3)}

\newcommand{\rotatingl}[3]{h(#1{:}#2{\rightarrow}#3)}
\newcommand{\CH}[1]{CH(#1)}

\newtheorem{lemma}{Lemma}

\newtheorem{conjecture}{Conjecture}
\newtheorem{theorem}{Theorem}
\newtheorem{observation}{Observation}
\newtheorem{proposition}{Proposition}
\newtheorem*{problem*}{Problem}

\begin{document}

\maketitle

\begin{abstract}
Counting interior-disjoint empty convex polygons in a point set is a typical {E}rd{\H{o}}s-{S}zekeres-type problem. We study this problem for 4-gons. Let $P$ be a set of $n$ points in the plane and in general position. A subset $Q$ of $P$, with four points, is called a $4$-{\em hole} in $P$ if $Q$ is in convex position and its convex hull does not contain any point of $P$ in its interior. 
Two 4-holes in $P$ are {\em compatible} if their interiors are disjoint. 
We show that $P$ contains at least $\lfloor 5n/11\rfloor {-} 1$ pairwise compatible 4-holes. This improves the lower bound of $2\lfloor(n-2)/5\rfloor$ which is implied by a result of Sakai and Urrutia (2007).
\end{abstract}

%\keywords{point set, convex quadrilateral, empty quadrilateral.}

\section{Introduction}
\label{introduction-section}
Throughout this paper, an $n$-{\em set} is a set of $n$ points in the plane and in general position, i.e., no three points are collinear. Let $P$ be an $n$-set. A {\em hole} in $P$ is a subset $Q$ of $P$, with at least three elements, such that $Q$ is in convex position and no element of $P$ lies in the interior of the convex hull of $Q$. A $k$-{\em hole} in $P$ is a hole with $k$ elements. By this definition, a 3-hole in $P$ is an empty triangle with vertices in $P$, and a 4-hole in $P$ is an empty convex quadrilateral with vertices in $P$.  

The problem of finding and counting holes in point sets has a long history in discrete combinatorial geometry, and has been an active research area since {E}rd{\H{o}}s and {S}zekeres~\cite{Erdos1978, Erdos1935} asked about the existence of $k$-holes in a point set.
In 1931, Esther Klein showed that any 5-set contains a convex quadrilateral~\cite{Erdos1935}; it is easy to see that it also contains a 4-hole. In 1978, Harborth~\cite{Harborth1978} proved that any 10-set contains a 5-hole. In 1983, Horton~\cite{Horton1983} exhibited arbitrarily large point sets with no 7-hole. The existence of a 6-hole in sufficiently large point sets has been proved by Nicol\'{a}s~\cite{Nicolas2007} and Gerken~\cite{Gerken2008}; a shorter proof of this result is given by Valtr~\cite{Valtr2008}. 

\begin{figure}[htb]
  \centering
\setlength{\tabcolsep}{0in}
  $\begin{tabular}{cc}
 \multicolumn{1}{m{.5\columnwidth}}{\centering\includegraphics[width=.25\columnwidth]{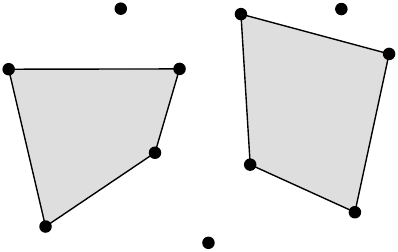}}
&\multicolumn{1}{m{.5\columnwidth}}{\centering\includegraphics[width=.25\columnwidth]{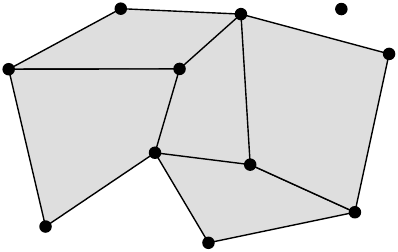}}
\end{tabular}$
  \caption{Two disjoint 4-holes (left), and five compatible 4-holes (right).}
\label{disjoint-compatible-fig}
\end{figure}

Two holes $Q_1$ and $Q_2$ are {\em disjoint} if their convex hulls are disjoint, i.e., they do not share any vertex and do not overlap. We say sat $Q_1$ and $Q_2$ are {\em compatible} if the interiors of their convex hulls are disjoint, that is, they can share vertices but do not overlap. A set of holes is called disjoint (resp. compatible) if its elements are pairwise disjoint (resp. compatible). See Figure~\ref{disjoint-compatible-fig}. 

Since every three points form the vertices of a triangle, by repeatedly creating a triangle with the three leftmost points of an $n$-set we obtain exactly $\lfloor n/3\rfloor$ disjoint 3-holes. However, this does not generalize to 4-holes, because the four leftmost points may not be in convex position. Obviously, the number of disjoint 4-holes in an $n$-set is at most $\lfloor n/4\rfloor$. Hosono and Urabe~\cite{Hosono2001} proved that the number of disjoint 4-holes is at least $\lfloor 5n/22\rfloor$; they improved this bound to $(3n{-}1)/13$ when $n=13\cdot 2^k {-}4$ for some $k\geqslant 0$. A variant of this problem where the 4-holes are vertex-disjoint, but can overlap, is considered in \cite{Wu2008}.
As for compatible holes, it is easy to verify that the number of compatible 3-holes in any $n$-set is at least $n{-}2$ and at most $2n{-}5$; these bounds are obtained by triangulating the point set: we get $n{-}2$ triangles, when the point set is in convex position, and $2n{-}5$ triangles, when the convex hull of the point set is a triangle. Sakai and Urrutia~\cite{Sakai2007} proved among other results that any 7-set contains at least two compatible 4-holes. In this paper we study the problem of finding the maximum number of compatible 4-holes in an $n$-set.  

Devillers~\etal~\cite{Devillers2003} considered some colored variants of this problem. They proved among other results that any bichromatic $n$-set has at least $\lceil n/4\rceil {-}2$ compatible monochromatic 3-holes; they also provided a matching upper bound. As for 4-holes, they conjectured that a sufficiently large bichromatic point set has a monochromatic 4-hole. Observe that any point set that disproves this conjecture does not have a 7-hole (regardless of colors). For a bichromatic point set $R\cup B$ in the plane, Sakai and Urrutia~\cite{Sakai2007} proved that if $|R|\geqslant 2|B|{+}5$, then there exists a monochromatic 4-hole. They also studied the problem of blocking 4-holes in a given point set $R$; the goal in this problem is to find a smallest point set $B$ such that any 4-hole in $R$ has a point of $B$ in its interior. The problem of blocking 5-holes has been studied by Cano~\etal~\cite{Cano2015}.

Aichholzer~\etal~\cite{Aichholzer2007} proved that every 11-set contains either a 6-hole, or a 5-hole and a disjoint 4-hole. Bhattacharya and Das~\cite{Bhattacharya2011} proved that every 12-set contains a 5-hole and a disjoint 4-hole. They also proved the existence of two disjoint 5-holes in every 19-set \cite{Bhattacharya2013}.
For more results on the number of $k$-holes in small point sets and other variations, the reader is referred to a paper by Aichholzer and Krasser~\cite{Aichholzer2001}, a summary of recent results by Aichholzer~\etal~\cite{Aichholzer2015}, and B. Vogtenhuber's doctoral thesis~\cite{Vogtenhuber2011}. Researchers also have studied the problem of counting the number of (not necessarily empty nor compatible) convex quadrilaterals in a point set; see, e.g.,  \cite{Aichholzer2006, Brodsky2003, Lovasz2004, Wagner2003}. 

A quadrangulation of a point set $P$ in the plane is a planar subdivision whose vertices are the points of $P$, whose outer face is the convex hull of $P$, and every internal face is a quadrilateral; in fact the quadrilaterals are empty and pairwise compatible. Similar to triangulations, quadrangulations have applications in finite element mesh generation, Geographic Information Systems (GIS), scattered data interpolation, etc.; see~\cite{Bose2002, Bose1997, Ramaswami1998, Toussaint1995}. Most of these applications look for a quadrangulation that has the maximum number of convex quadrilaterals. To maximize the number of convex quadrilaterals, various heuristics and experimental results are presented in \cite{Bose2002, Bose1997}. This raises another motivation to study theoretical aspects of compatible empty convex quadrilaterals in a planar point set.

In this paper we study lower and upper bounds for the number of compatible 4-holes in point sets in the plane. A trivial upper bound is $\lfloor n/2\rfloor -1$ which comes from $n$ points in convex position. The $\lfloor 5n/22\rfloor$ lower bound on the number of disjoint 4-holes that is proved by Hosono and Urabe~\cite{Hosono2001}, simply carries over to the number of compatible 4-holes. Also, as we will see in Section~\ref{preliminary-section}, the lower bound of $2\lfloor(n-2)/5\rfloor$ on the number of compatible 4-holes is implied by a result of Sakai and Urrutia \cite{Sakai2007}. After some new results for small point sets, we prove non-trivial
lower bounds on the number of compatible 4-holes in an $n$-set.  In Section~\ref{preliminary-section} we introduce some notations and prove some preliminary results. In Section~\ref{9-11-section} we prove that every 9-set (resp. 11-set) contains three (resp. four) compatible 4-holes. Using these results, in Section~\ref{n-section}, we prove that every $n$-set contains at least $\lfloor 5n/11\rfloor {-} 1$ compatible 4-holes. Our proof of this lower bound is constructive, and  immediately yields an $O(n\log ^2 n)$-time algorithm for finding this many compatible 4-holes. 
\removed{In summary, the real bound is conjectured to be roughly $0.5 n$, while the best previously known lower bound is roughly $0.4 n$. We prove a roughly $0.454 n$ lower bound, which is a considerable progress towards closing the gap.}

Since the initial presentation of this work \cite{Biniaz2017}, the problem has attracted further attention. Most prominently, the lower bound on the number of compatible 4-holes has been improved to $\lceil\frac{n-3}{2}\rceil$ by Cravioto-Lagos, Gonz\'{a}lez-Mart\'{\i}nez, Sakai, and Urrutia \cite{Urrutia2017}. The same bound is claimed in an abstract by Lomeli-Haro, Sakai, and Urrutia in Kyoto International Conference on Computational Geometry and Graph Theory (KyotoCGGT2007) \cite{Lomeli2007}. However, this result has not been published yet. 
\section{Preliminaries}
\label{preliminary-section}
\begin{wrapfigure}{r}{1.8in} 
%\vspace{-10pt} 
\includegraphics[width=1.8in]{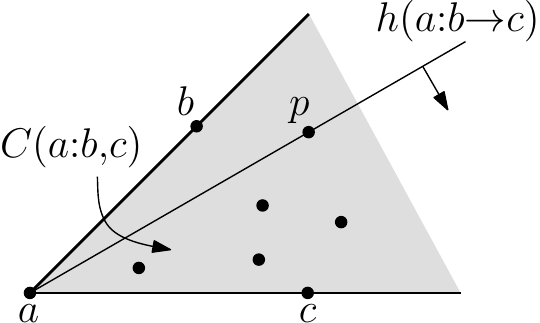} 
\vspace{-20pt} 
\end{wrapfigure}
First we introduce some notation from \cite{Hosono2001}.
We define the {\em convex cone} $\cone{a}{b}{c}$ to be the region of the angular domain in the plane that is determined by three non-collinear points $a$, $b$, and $c$, where $a$ is the apex, $b$ and $c$ are on the boundary of the domain, and $\angle bac$ is acute (less than $\pi/2$). 
%We denote by $\icone{a}{b}{c}$ the {\em image} of $\cone{a}{b}{c}$. 
We denote by $\rotatingl{a}{b}{c}$ the rotated half-line that is anchored at $a$ and rotates, in $\cone{a}{b}{c}$, from the half-line $ab$ to the half-line $ac$. If the interior of $\cone{a}{b}{c}$ contains some points of a given point set,
then we call the first point that $\rotatingl{a}{b}{c}$ meets the {\em attack point} of $\rotatingl{a}{b}{c}$; the point $p$ in the figure to the right is the attack point. 
%We use the notation $ab$ to refer to the straight line through $a$ and $b$.

Let $P$ be an $n$-set. We denote by $\CH{P}$ the convex hull of $P$. Let $p_0$ be the bottommost vertex on $\CH{P}$. Without loss of generality assume that $p_0$ is the origin. Label the other points of $P$ by $p_1,\dots,p_{n-1}$ in clockwise order around $p_0$, starting from the negative $x$-axis; see Figure~\ref{radial-fig}(a). We refer to the sequence $p_1,\dots,p_{n-1}$ as the {\em radial ordering} of the points of $P\setminus\{p_0\}$ around $p_0$. We denote by $l_{i,j}$ the straight line through two points with indexed labels $p_i$ and $p_j$.

\begin{figure}[htb]
  \centering
\setlength{\tabcolsep}{0in}
  $\begin{tabular}{cc}
 \multicolumn{1}{m{.6\columnwidth}}{\centering\includegraphics[width=.48\columnwidth]{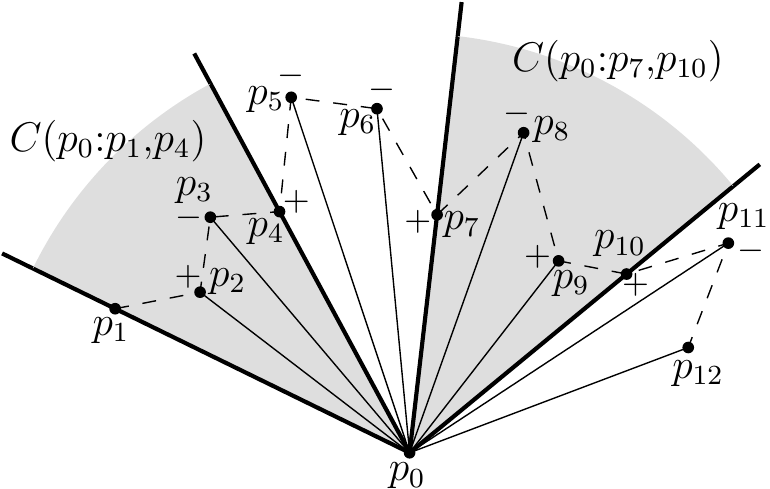}}
&\multicolumn{1}{m{.4\columnwidth}}{\centering\includegraphics[width=.25\columnwidth]{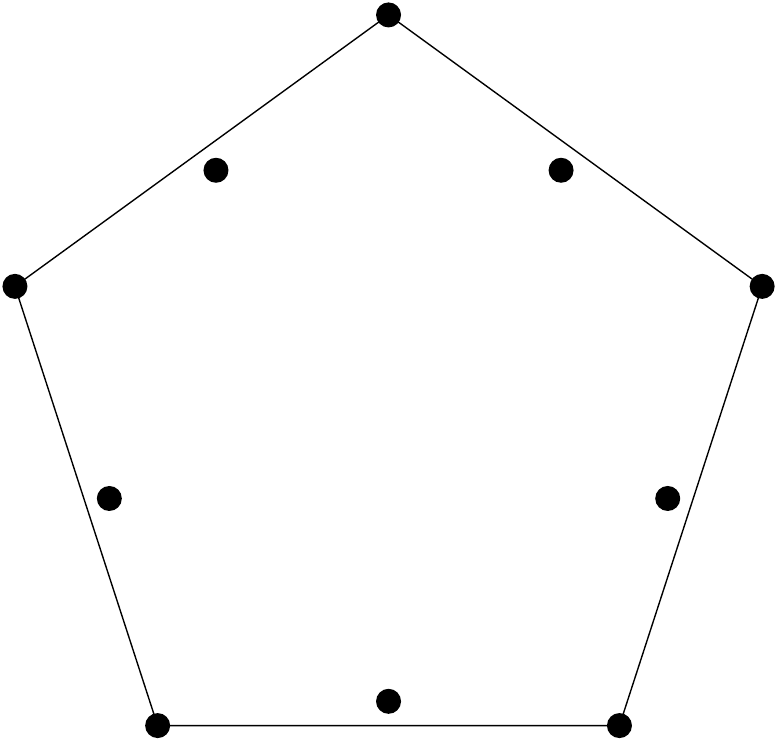}}
\\
(a) & (b)
\end{tabular}$
  \caption{(a) The radial ordering of points around $p_0$. (b) A 10-set with at most three compatible 4-holes.}
\label{radial-fig}
\end{figure}

It is easy to verify that the number of 4-holes in an $n$-set in convex position is exactly $\lfloor n/2\rfloor{-}1$. Figure~\ref{radial-fig}(b), that is borrowed from ~\cite{Hosono2001}, shows an example of a 10-set that contains at most three compatible 4-holes; by removing a vertex from the convex hull, we obtain a 9-set with the same number of 4-holes. This example can be extended to larger point sets, and thus, to the following proposition.

\begin{proposition}
\label{upper-bound-prop}
For every $n\geqslant 3$, there exists an $n$-set that has at most $\lceil n/2\rceil{-}2$ compatible 4-holes.
\end{proposition}

\begin{proposition}
The number of compatible 4-holes in an $n$-set is at most $n-3$.
\end{proposition}
\begin{proof}
Let $P$ be an $n$-set. Consider the maximum number of compatible 4-holes in $P$. The point set $P$ together with an edge set, that is the union of the boundary edges of these 4-holes, introduces a planar graph $G$. Every 4-hole in $P$ corresponds to a 4-face (a face with four edges) in $G$, and vice versa. Using Euler formula for planar graphs one can verify that the number of internal 4-faces of $G$ is at most $n-3$. This implies that the number of 4-holes in $P$ is also at most $n-3$. 
\end{proof}

\begin{theorem}[Klein; see~\cite{Erdos1935}]
\label{Klein-thr}
Every 5-set contains a 4-hole.
\end{theorem}

\begin{theorem}[Sakai and Urrutia~\cite{Sakai2007}]
\label{Sakai-thr}
Every 7-set contains at least two compatible 4-holes.
\end{theorem}

As a warm-up exercise, we show that the number of 4-holes in an $n$-set $P$ is at least $\lfloor (n-2)/3 \rfloor$. Let $p_0$ be the bottommost point of $P$ and let $p_1,\dots,p_{n-1}$ be the radial ordering of the other points of $P$ around $p_0$. Consider $\lfloor (n-2)/3 \rfloor$ cones $\cone{p_0}{p_1}{p_4}, \cone{p_0}{p_4}{p_7}, \cone{p_0}{p_7}{p_{10}},\dots$ where each cone has three points of $P$ (including $p_0$) on its boundary and two other points in its interior. See Figure~\ref{radial-fig}(a). Each cone contains five points (including the three points on its boundary), and by Theorem~\ref{Klein-thr} these five points introduce a 4-hole. Since the interiors of these cones are pairwise disjoint, we get $\lfloor (n-2)/3 \rfloor$ compatible 4-holes in $P$. We can improve this bound as follows. By defining the cones as $\cone{p_0}{p_1}{p_6}, \cone{p_0}{p_6}{p_{11}}, \cone{p_0}{p_{11}}{p_{16}},\dots$, we get $\lfloor (n-2)/5 \rfloor$ cones, each of which contains seven points. By Theorem~\ref{Sakai-thr}, the seven points in each cone introduce  two compatible 4-holes, and thus, we get $2\cdot\lfloor (n-2)/5 \rfloor$ compatible 4-holes in total. Intuitively, any improvement on the lower bound for small point sets carries over to large point sets.

\begin{lemma}
\label{six-lemma}
Any 6-set, that has five or six vertices on the boundary of its convex hull, contains two compatible 4-holes.
\end{lemma}
\begin{proof}
Let $P=\{p_0,\dots,p_5\}$ be a 6-set with five or six vertices on $\CH{P}$. If $\CH{P}$ has six vertices, then $P$ is in convex position, and thus, contains two compatible 4-holes. Assume that $\CH{P}$ has five vertices. Also, without loss of generality, assume that $p_5$ is in the interior of $\CH{P}$ and that $p_0,\dots, p_4$ is the clockwise order of the vertices of $\CH{P}$. Consider five triangles $\bigtriangleup p_ip_{i+1}p_{i+3}$ with $i\in\{0,\dots, 4\}$; all indices are modulo 5. The union of these triangles cover the interior of $\CH{P}$. Thus, $p_5$ lies in a triangle $\bigtriangleup p_ip_{i+1}p_{i+3}$ for some $i\in\{0,\dots, 4\}$. Therefore, the two quadrilaterals $p_5p_{i+1}p_{i+2}p_{i+3}$ and $p_5p_{i+3}p_{i+4}p_{i}$ are empty, convex, and internally disjoint.
\end{proof}

\section{Compatible 4-holes in 9-sets and 11-sets}
\label{9-11-section}

In this section we provide lower bounds on the number of compatible 4-holes in 9-sets and 11-sets. In Subsection~\ref{9-section} we prove that every 9-set contains at least three compatible 4-holes. In Subsection~\ref{11-section} we prove that every 11-set contains at least four compatible 4-holes. Both of these lower bounds match the upper bounds given in Proposition~\ref{upper-bound-prop}. Due to the nature of this type of problems, our proofs involve case analysis. The case analysis gets more complicated as the number of points increases. 
To simplify the case analysis, we use two observations and a lemma, that are given in Subsection~\ref{obs-lemma-section}, to find 4-holes. To simplify the case analysis further, we prove our claim for 9-sets first, then we use this result to obtain the proof for 11-sets. In this section we may use the term ``quadrilateral'' instead of 4-hole.

Let $P$ be an $n$-set. Let $p_0$ be the bottommost point of $P$ and let $p_1,\dots, p_{n-1}$ be the radial ordering of the other points of $P$ around $p_0$. For each point $p_i$, with $i\in\{2,\dots, n-2\}$, we define the {\em signature} $s(p_i)$ of $p_i$ to be ``$+$'' if, in the quadrilateral $p_0p_{i-1} p_i p_{i+1}$, the inner angle at $p_i$ is greater than $\pi$, and ``$-$'' otherwise; see Figure~\ref{radial-fig}(a). We refer to $s(p_2)s(p_3)\dots s(p_{n-2})$ as the {\em signature sequence} of $P$ with respect to $p_0$. We refer to $s(p_{n-2})\dots s(p_3)s(p_2)$ as the {\em reverse} of $s(p_2)s(p_3)\dots s(p_{n-2})$. A {\em minus subsequence} is a subsequence of $-$ signs in a signature sequence. A {\em plus subsequence} is defined analogously. For a given signature sequence $\delta$, we denote by $m(\delta)$, the number of minus signs in $\delta$.

\begin{figure}[htb]
  \centering
\setlength{\tabcolsep}{0in}
  $\begin{tabular}{cc}
 \multicolumn{1}{m{.5\columnwidth}}{\centering\includegraphics[width=.28\columnwidth]{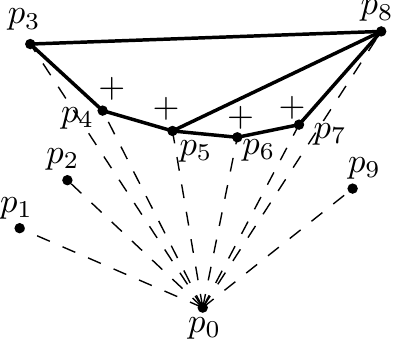}}
&\multicolumn{1}{m{.5\columnwidth}}{\centering\includegraphics[width=.32\columnwidth]{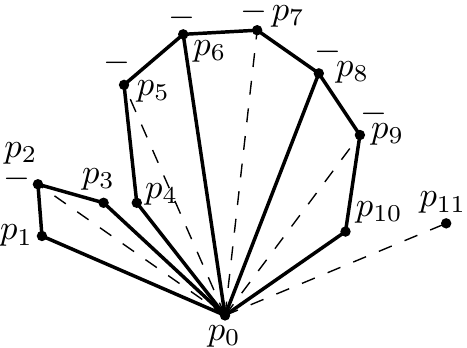}}
\\
(a) & (b)
\end{tabular}$
  \caption{(a) A plus subsequence $s(p_4)s(p_5)s(p_6)s(p_7)$ of length four. (b) Two minus subsequences $s(p_2)$ and $s(p_5)\dots s(p_9)$ of lengths one and five.}
\label{signature-fig}
\end{figure}

\subsection{Two observations and a lemma}
\label{obs-lemma-section}
In this section we introduce two observations and a lemma to simplify some case analysis in our proofs, which come later. Notice that if $s(p_{i})\dots s(p_{j})$ is a plus subsequence, then the points $p_{i-1},p_i,\dots,p_j,p_{j+1}$ are in convex position and the interior of their convex hull does not contain any point of $P$. Also, if $s(p_{i})\dots s(p_{j})$ is a minus subsequence, then the points $p_0,p_{i-1},p_i,\dots,p_j,p_{j+1}$ are in convex position and the interior of their convex hull does not contain any point of $P$. Therefore, the following two observations are valid.

\begin{observation}
\label{even-plus-obs}
Let $s(p_i)\dots s(p_j)$ be a plus subsequence of length $2k$, with $k\geqslant 1$. Then, the convex hull of $p_{i-1},\dots, p_{j+1}$ can be partitioned into $k$ compatible 4-holes. See Figure~\ref{signature-fig}$($a$)$.
\end{observation}

\begin{observation}
\label{odd-minus-obs}
Let $s(p_i)\dots s(p_j)$ be a minus subsequence of length $2k+1$, with $k\geqslant 0$. Then, the convex hull of $p_{0},p_{i-1},\dots, p_{j+1}$ can be partitioned into $k+1$ compatible 4-holes. See Figure~\ref{signature-fig}$($b$)$.
\end{observation}

\iffalse
\begin{observation}
\label{even-plus-obs}
Let $s(p_{i+1})s(p_{i+2})\dots s(p_{i+2k})$ be an even-plus subsequence, with $k\geqslant 1$. Then, the convex hull of $p_{i},p_{i+1},\dots, p_{i+2k+1}$ can be partitioned into $k$ compatible 4-holes. See Figure~\ref{signature-fig}$($a$)$.
\end{observation}

\begin{observation}
\label{odd-minus-obs}
Let $s(p_{i+1})s(p_{i+2})\dots s(p_{i+2k-1})$ be an odd-minus subsequence, with $k\geqslant 1$. Then, the convex hull of $p_{0},p_i,\dots, p_{i+2k}$ can be partitioned into $k+1$ compatible 4-holes. See Figure~\ref{signature-fig}$($b$)$.
\end{observation}
\fi

\begin{lemma}
\label{even-minus-lemma}
Let $s(p_{i+1})s(p_{i+2})\dots s(p_{i+2k})$ be a minus subsequence of length $2k$, with $k\geqslant 1$, and let $p_i$ and $p_{i+2k+1}$ have $+$ signatures. Then, one can find $k+1$ compatible 4-holes in the convex hull of $p_{0},p_{i-1},\dots,\allowbreak p_{i+2k+2}$.
\end{lemma}
\begin{proof}
Refer to Figure~\ref{even-minus-fig}. For every $j\in\{0,\dots, k\}$ let $l_{i+j}$ be the line through $p_{i+j}$ and $p_{i+2k+1-j}$. These lines might intersect each other, but, for a better understanding of this proof, we visualized them as parallel lines in Figure~\ref{even-minus-fig}. 

\begin{figure}[htb]
  \centering
\setlength{\tabcolsep}{0in}
  $\begin{tabular}{cc}
 \multicolumn{1}{m{.5\columnwidth}}{\centering\includegraphics[width=.45\columnwidth]{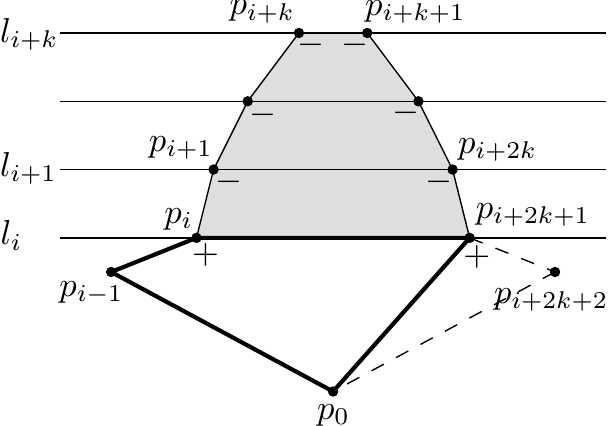}}
&\multicolumn{1}{m{.5\columnwidth}}{\centering\includegraphics[width=.45\columnwidth]{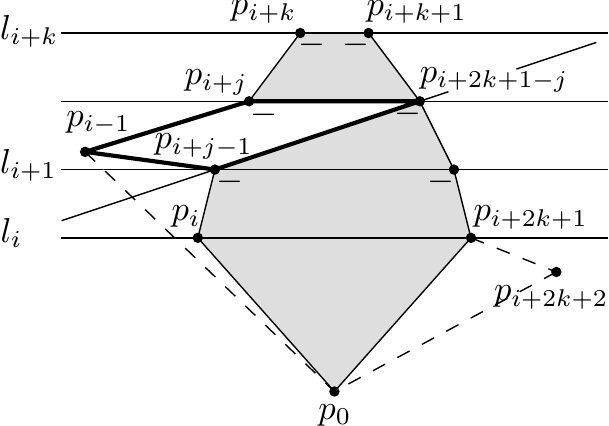}}
\\
(a) & (b)
\end{tabular}$
  \caption{The point $p_{i-1}$ is (a) below $l_{i}$, and (b) below $l_{i+j}$ and above all lines $l_i,\dots, l_{i+j-1}$.}
\label{even-minus-fig}
\end{figure}

Notice that the points $p_0, p_i,\dots, p_{i+2k+1}$ are in convex position. If $p_{i-1}$ is below $l_{i}$, then we get a 4-hole $p_0p_{i-1}p_{i}p_{i+2k+1}$ and $k$ other compatible 4-holes in the convex hull of the points $p_i,\dots,p_{i+2k+1}$; see Figure~\ref{even-minus-fig}(a). Assume $p_{i-1}$ is above $l_{i}$. If $p_{i-1}$ is below some lines in the sequence $l_{i+1},\dots,l_{i+k}$, then let $l_{i+j}$ be the first one in this sequence, that is, $p_{i-1}$ is below $l_{i+j}$ but above all lines $l_i,\dots, l_{i+j-1}$. Notice that in this case $p_{i-1}$ is also above the line through $p_{i+j-1}$ and $p_{i+2k+1-j}$. In this case we get a 4-hole $p_{i-1}p_{i+j}p_{i+2k+1-j}p_{i+j-1}$, and $k-j$ compatible 4-holes in the convex hull of $p_{i+j}\dots,p_{i+2k+1-j}$, and $j$ compatible 4-holes in the convex hull of $p_0,p_i,\dots,p_{i+j-1},p_{i+2k+1-j},\dots,p_{i+2k+1}$; see Figure~\ref{even-minus-fig}(b). Thus, we get $k+1$ compatible 4-holes in total. Similarly, if $p_{i+2k+2}$ is below one of the lines $l_{i+j}$ for $j\in\{0,\dots,k\}$ we get $k+1$ compatible 4-holes. Thus, assume that both $p_{i-1}$ and $p_{i+2k+2}$ are above all lines $l_i,\dots,l_{i+k}$. In this case we get a 4-hole $p_{i-1}p_{i+2k+2}p_{i+k+1}p_{i+k}$ and $k$ other compatible 4-holes in the convex hull of $p_i,\dots, p_{i+2k+1}$. Thus, we get $k+1$ compatible 4-holes in total. 
\end{proof}

Notice that the statement of Lemma~\ref{even-minus-lemma} is true regardless of the signatures of $p_i$ and $p_{i+2k+1}$. However, in this paper, when we apply this lemma, $p_i$ and $p_{i+2k+1}$ have $+$ signatures. 

Quadrilaterals obtained by Observations~\ref{even-plus-obs} and \ref{odd-minus-obs} do not overlap because quadrilaterals obtained by Observation~\ref{even-plus-obs} lie above the chain $p_1,\dots,p_{n-1}$ while quadrilaterals obtained by Observation~\ref{odd-minus-obs} lie below this chain. However, the quadrilaterals obtained in the proof of Lemma~\ref{even-minus-lemma} might lie above and/or below this chain. The quadrilaterals obtained by this lemma overlap the quadrilaterals obtained by Observations~\ref{even-plus-obs} or \ref{odd-minus-obs} in the following two cases:

\begin{itemize}
 \item Consider the first case in the proof of Lemma~\ref{even-minus-lemma} when $p_{i-1}$ lies below $l_{i}$ and we create the quadrilateral $p_0p_{i-1}p_{i}p_{i+2k+1}$. If $s(p_{i-1})$ belongs to a minus subsequence, and we apply Observation~\ref{odd-minus-obs} on it, then the quadrilateral $p_0p_{i-2}p_{i-1}p_i$ obtained by this observation overlaps the quadrilateral $p_0p_{i-1}p_{i}p_{i+2k+1}$. Similar issue may arise when $s(p_{i+2k+2})$ belongs to a minus subsequence. 
\item Consider the last two cases in the proof of Lemma~\ref{even-minus-lemma} when $p_{i-1}$ lies above $l_{i}$. If $s(p_{i-1})$ belongs to a plus subsequence, and we apply Observation~\ref{even-plus-obs} on it, then the quadrilaterals obtained by this observation might overlap either the quadrilateral $p_{i-1}p_{i+j}p_{i+2k+1-j}p_{i+j-1}$ or the quadrilateral $p_{i-1}p_{i+2k+2}p_{i+k+1}p_{i+k}$ that is obtained by  Lemma~\ref{even-minus-lemma}. Similar issue may arise when $s(p_{i+2k+2})$ belongs to a plus subsequence. 
\end{itemize}

As such, in our proofs, we keep track of the following two assertions when applying Lemma~\ref{even-minus-lemma} on a subsequence $s(p_{i+1})s(p_{i+2})\dots s(p_{i+2k})$:

\begin{enumerate}[leftmargin=3cm, label={\bfseries{Assertion {\arabic*}.}}]
 \item 
 Do not apply Observation~\ref{even-plus-obs} on a plus subsequence that contains $s(p_{i-1})$ or $s(p_{i+2k+2})$.
 \item Do not apply Observation~\ref{odd-minus-obs} on a minus subsequence that contains $s(p_{i-1})$ or $s(p_{i+2k+2})$. 
\end{enumerate}

\subsection{Three quadrilaterals in 9-sets}
\label{9-section}
In this section we prove our claim for 9-sets:

\begin{theorem}
\label{9-thr}
Every 9-set contains at least three compatible 4-holes.
\end{theorem}

Let $P$ be a 9-set. Let $p_0$ be the bottommost point of $P$ and let $p_1,\dots,p_{8}$ be the radial
ordering of the other points of $P$ around $p_0$. Let $\delta$ be the signature sequence of $P$ with respect to $p_0$, i.e., $\delta=s(p_2)\dots s(p_6)s(p_7)$. Depending on the value of $m(\delta)$, i.e., the number of minus signs in $\delta$, we consider the following seven cases. Notice that any proof of this theorem for $\delta$ carries over to the reverse of $\delta$ as well. So, in the proof of this theorem, if we describe a solution for a signature sequence, we skip the description for its reverse. 

\begin{itemize}
 \item $m(\delta)=0$: In this case $\delta$ is a plus subsequence of length six. Our result follows by Observation~\ref{even-plus-obs}.
 \item $m(\delta)=1$: In this case $\delta$ has five plus signs. By Observation~\ref{odd-minus-obs}, we get a quadrilateral by the point with $-$ signature. If four of the plus signs are consecutive, then by Observation~\ref{even-plus-obs} we get two more quadrilaterals. Otherwise, $\delta$ has two disjoint subsequences of plus signs, each of length at least two. Again, by Observation~\ref{even-plus-obs} we get a quadrilateral for each of these subsequences. Therefore, in total we get three 4-holes; observe that these 4-holes are pairwise non-overlapping.
 \item $m(\delta)=2$: Notice that $\delta$ has a plus subsequence of length at least two. If the two minus signs are non-consecutive, then we get two quadrilaterals by Observation~\ref{odd-minus-obs} and one by Observation~\ref{even-plus-obs}. Assume the two minus signs are consecutive. If the four plus signs are consecutive or partitioned into two subsequences of lengths two, then we get two quadrilaterals by Observation~\ref{even-plus-obs} and one by Observation~\ref{odd-minus-obs}. The remaining sequences are $+--+++$ and $+++--+$, where the second sequence is the reverse of the first one. By splitting the first sequence as $+--+|++$ we get two quadrilaterals for the subsequence $+--+$, by Lemma~\ref{even-minus-lemma}. If in this lemma we land up in the last case where both $p_{i-1}$ and $p_{i+2k+2}$ are above $l_{i+k}$, then we get a third compatible quadrilateral $p_1p_6p_7p_8$, otherwise we get $p_4p_6p_7p_8$. Notice that Assertion 1 holds here.
 \item $m(\delta)=3$: If the three minus signs are pairwise non-consecutive, then we get three quadrilaterals by Observation~\ref{odd-minus-obs}. If the three minus signs are consecutive, then $\delta$ has a plus subsequence of length at least two. Thus, we get two quadrilaterals by Observation~\ref{odd-minus-obs} and one by Observation~\ref{even-plus-obs}. Assume the minus signs are partitioned into two disjoint subsequences of lengths one and two. Then, we get two quadrilaterals for the minus signs. If $\delta$ has a plus subsequence of length at least two, then we get a third quadrilateral by this subsequence. The remaining sequences are $+--+-+$ and its reverse. 

\begin{figure}[htb]
  \centering
\setlength{\tabcolsep}{0in}
  $\begin{tabular}{cc}
 \multicolumn{1}{m{.52\columnwidth}}{\centering\includegraphics[width=.45\columnwidth]{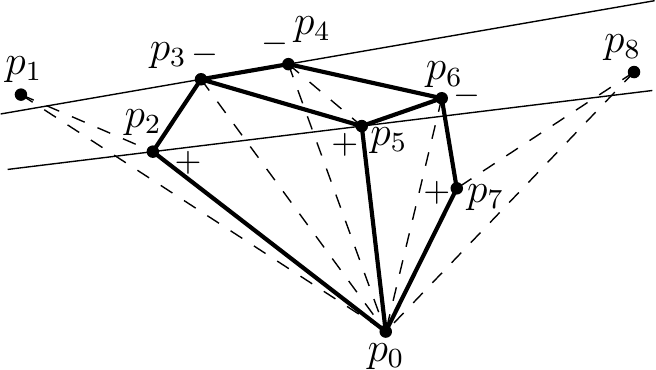}}
&\multicolumn{1}{m{.48\columnwidth}}{\centering\includegraphics[width=.42\columnwidth]{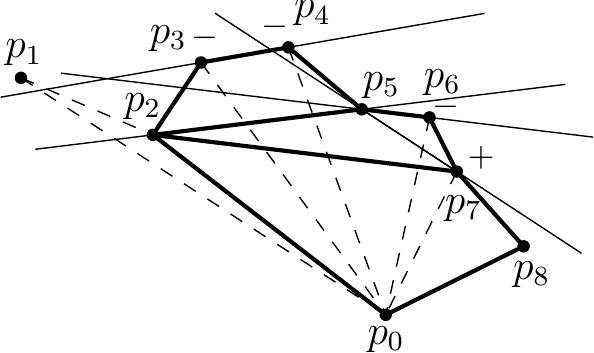}}
\\
(a) & (b)
\end{tabular}$
  \caption{Signature sequence $+--+-+$. (a) $p_1$ is above $l_{3,4}$, and $p_6$ is below $l_{3,4}$ and above $l_{2,5}$. (b) $p_1$ is above $l_{3,4}$, $p_6$ is below $l_{2,5}$, and $p_8$ is below $l_{5,7}$.}
\label{9-3m-fig}
\end{figure}

We show how to get three compatible 4-holes with the sequence $+--+-+$. See Figure~\ref{9-3m-fig}. First we look at $p_1$. If $p_1$ is below $l_{2,5}$ then the three quadrilaterals $p_0p_1p_2p_5$, $p_2p_3p_4p_5$, and $p_0p_5p_6p_7$ are compatible. Assume $p_1$ is above $l_{2,5}$. If $p_1$ is below $l_{3,4}$ then the quadrilaterals $p_1p_3p_4p_2$, $p_0p_2p_4p_5$, and $p_0p_5p_6p_7$ are compatible. Assume $p_1$ is above $l_{3,4}$. Now, we look at $p_6$. If $p_6$ is above $l_{3,4}$ then $p_1p_6p_4p_3$, $p_2p_3p_4p_5$, and $p_0p_5p_6p_7$ are compatible. If $p_6$ is below $l_{3,4}$ and above $l_{2,5}$ as in Figure~\ref{9-3m-fig}(a), then $p_0p_2p_3p_5$, $p_3p_4p_6p_5$, and $p_0p_5p_6p_7$ are compatible. Assume $p_6$ is below $l_{2,5}$ as in Figure~\ref{9-3m-fig}(b); consequently $p_7$ is also below $l_{2,5}$ because $p_6$ has $-$ signature. Since $p_5$ has $+$ signature, $p_4$ is above $l_{5,6}$. Now, we look at $p_8$. If $p_8$ is above $l_{5,6}$, then $p_4p_8p_6p_5$, $p_2p_3p_4p_5$, and $p_0p_5p_6p_7$ are compatible. If $p_8$ is below $l_{5,6}$ and above $l_{5,7}$, then $p_5p_6p_8p_7$, $p_2p_3p_4p_5$, and $p_0p_2p_5p_7$ are compatible. Assume $p_8$ is below $l_{5,7}$ as in Figure~\ref{9-3m-fig}(b). In this case $p_2p_3p_4p_5$, $p_2p_5p_6p_7$, and $p_0p_2p_7p_8$ are compatible. 

 \item $m(\delta)=4$: If the two plus signs in $\delta$ are consecutive, then we get one quadrilateral by Observation~\ref{even-plus-obs} and two by Observation~\ref{odd-minus-obs}. Assume the two plus signs are non-consecutive. If the minus signs are partitioned into three subsequences or two subsequences of lengths one and three, then we get three compatible 4-holes by Observation~\ref{odd-minus-obs}. The remaining sequences are $+----+$, $+--+--$ and its reverse. For the sequence $+----+$ we get three quadrilaterals by Lemma~\ref{even-minus-lemma}. The sequence $+--+--$ can be handled by splitting as $+--+|-|-$, where we get two quadrilaterals for the subsequence $+--+$, by Lemma~\ref{even-minus-lemma}, and one quadrilateral for the last minus sign, by Observation~\ref{even-plus-obs}. Notice that Assertion 2 holds here as we apply Observation~\ref{even-plus-obs} on the last minus sign. 

 \item $m(\delta)=5$: If the five minus signs are consecutive, then we get three compatible quadrilaterals by Observation~\ref{odd-minus-obs}. Otherwise, $\delta$ has two minus subsequences, one of which has size at least three. Again, by Observation~\ref{odd-minus-obs} we get three quadrilaterals with these two subsequences. 
 \item $m(\delta)=6$: The six minus signs are consecutive and our result follows by Observation~\ref{odd-minus-obs}.
\end{itemize}

\begin{wrapfigure}{r}{1.2in} 
\vspace{-10pt} 
\centering
\includegraphics[width=1in]{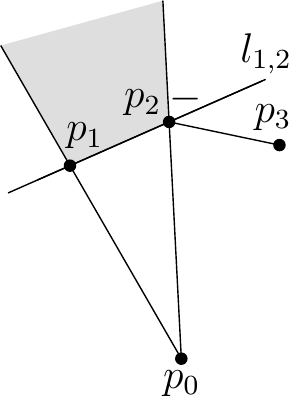} 
%\vspace{-20pt} 
\end{wrapfigure}
This finishes the proof of Theorem~\ref{9-thr}. We will refer to the following remark in our proof for 11-sets.\\

\vspace{-5pt}
\noindent{\bf Remark 1:}
When $\delta$ starts with a $-$ sign, the proof of Theorem~\ref{9-thr} does not connect $p_1$ to any point above $l_{1,2}$. That is, in the cone $\cone{p_0}{p_1}{p_2}$, the region that is above $l_{1,2}$ (the shaded region in the figure to the right), is disjoint from the interiors of the three quadrilaterals obtained in the proof of this theorem. An analogous argument is valid when $\delta$ ends with a $-$ sign.

\subsection{Four quadrilaterals in 11-sets}
\label{11-section}

In this section we prove our claim for 11-sets:

\begin{theorem}
\label{11-thr}
Every 11-set contains at least four compatible 4-holes.
\end{theorem}

Let $P$ be an 11-set. Let $p_0$ be the bottommost point of $P$ and let $p_1,\dots,p_{10}$ be the radial
ordering of the other points of $P$ around $p_0$. Let $\delta=s(p_2)\dots s(p_9)$ be the signature sequence of $P$ with respect to $p_0$. Depending on the value of $m(\delta)$ we will have nine cases. As in the proof of Theorem~\ref{9-thr}, if we describe a solution for a signature sequence, we skip the description for its reverse. 

Assume $\delta$ starts with a $-$ signature. Let $P'=P\setminus\{p_1,p_2\}$, and notice that $P'$ has nine points. By Theorem~\ref{9-thr} we get three compatible quadrilaterals with points of $P'$. We get $p_0p_1p_2p_3$ as the fourth quadrilateral; notice that this quadrilateral does not overlap any of the quadrilaterals that are obtained from $P'$. Thus our result follows. Similarly, if $\delta$ ends with $-$, we get four compatible quadrilaterals. Therefore, in the rest of the proof we assume that $\delta$ starts and ends with plus signs. Because of this, we will not have the cases where $m(\delta)\in\{7,8\}$, and thus, we describe the remaining cases where $m(\delta)\in\{0,\dots,6\}$

Assume $\delta$ starts with the $++-$ subsequence. Let $P'=P\setminus\{p_1,p_2\}$. Let $p'_0,\dots,p'_8$ be the corresponding labeling of points in $P'$ where $p'_0=p_0$, $p'_1=p_3$, $p'_2=p_4$, and so on. By applying Theorem~\ref{9-thr} on $P'$, we get three compatible quadrilaterals $Q_1$, $Q_2$, and $Q_3$. We get the fourth quadrilateral by $Q_4=p_1p_2p_3p_4$; we have to make sure that $Q_4$ does not overlap any of $Q_1$, $Q_2$, and $Q_3$. The signature sequence of $P'$ starts with $-$, i.e., $p'_2$ has minus signature. By Remark 1, in the cone $\cone{p'_0}{p'_1}{p'_2}$, the region that is above the line through $p'_1$ and $p'_2$, is disjoint from the interiors of $Q_1$, $Q_2$, and $Q_3$. Thus, $Q_4$ is compatible with $Q_1$, $Q_2$, $Q_3$, and our result follows. Similarly, if $\delta$ ends with $-++$, then we obtain four compatible  quadrilaterals. Therefore, in the rest of this proof we assume that $\delta$ starts with $+-$ or $+++$, and ends with $-+$ or $+++$. 

\begin{itemize}
 \item $m(\delta)=0$: In this case $\delta$ is a plus subsequence of length eight. Our result follows by Observation~\ref{even-plus-obs}.
 \item $m(\delta)=1$: In this case $\delta$ has seven plus signs. By Observation~\ref{odd-minus-obs} we get one quadrilateral by the point with $-$ signature. If six of the plus signs are consecutive, then by Observation~\ref{even-plus-obs} we get three more quadrilaterals. Otherwise, $\delta$ has two disjoint subsequences of plus signs, one of which has length at least two and the other has length at least four. Again, by Observation~\ref{even-plus-obs} we get one quadrilateral from the first subsequence and two from the second subsequence. Therefore, we get four compatible quadrilaterals in total.
 \item $m(\delta)=2$: Assume the two minus signs are non-consecutive. Then, we get two quadrilaterals by Observation~\ref{odd-minus-obs}. Moreover, $\delta$ has either one plus subsequence of length at least four or two plus subsequences each of length at least two. In either case, we get two other quadrilaterals by Observation~\ref{even-plus-obs}. Assume the two minus signs are consecutive. Then, we get one quadrilateral by Observation~\ref{odd-minus-obs}. If the six plus signs are consecutive or partitioned into two subsequences of lengths two and four, then we get three other quadrilaterals by Observation~\ref{even-plus-obs}. The remaining cases are $+--+++++$, $+++--+++$ and their reverses. For the first sequence, by splitting it as $+--+|++++$ we get two quadrilaterals for the subsequence $+--+$, by Lemma~\ref{even-minus-lemma}. We get the third and fourth compatible quadrilaterals as follows: If in Lemma~\ref{even-minus-lemma} we land up in the last case where both $p_{i-1}$ and $p_{i+2k+2}$ are above $l_{i+k}$, then we get $p_1p_6p_7p_8$ and $p_1p_8p_9p_{10}$, otherwise we get $p_4p_6p_7p_8$ and $p_4p_8p_9p_{10}$. For the second sequence, by splitting it as $++|+--+|++$ we get two quadrilaterals for the subsequence $+--+$, by Lemma~\ref{even-minus-lemma}. We get the third and fourth compatible quadrilaterals as follows: If in Lemma~\ref{even-minus-lemma} we land up in the last case where both $p_{i-1}$ and $p_{i+2k+2}$ are above $l_{i+k}$, then we get $p_1p_2p_3p_8$ and $p_1p_8p_9p_{10}$, otherwise we get $p_1p_2p_3p_5$ and $p_6p_8p_9p_{10}$. Notice that Assertion 1 holds in both cases. 

 \item $m(\delta)=3$: If the three minus signs are pairwise non-consecutive, then $\delta$ has a plus subsequence of length at least two. Thus, we get three quadrilaterals by Observation~\ref{odd-minus-obs} and one by Observation~\ref{even-plus-obs}. If the three minus signs are consecutive, then $\delta$ either has a plus subsequence of length at least four or two plus subsequences each of length at least two. In either case, we get two quadrilaterals by Observation~\ref{odd-minus-obs} and two by Observation~\ref{even-plus-obs}. Thus, we assume that the minus signs are partitioned into two disjoint subsequences of lengths one and two. Then, we can get two quadrilaterals for the minus signs. If $\delta$ has a plus subsequence of length four or two disjoint plus subsequences each of length at least two, then we get two other quadrilaterals by the plus signs. Assuming otherwise, the remaining subsequences are $+-+++--+$, $+-+--+++$, $+--+-+++$, and their reverses. For the first sequence, by spitting it as $+|-|++|+--+$ we get one quadrilateral for the subsequence $-$, by Observation~\ref{odd-minus-obs}, and two quadrilaterals for the subsequence $+--+$, by Lemma~\ref{even-minus-lemma}. If in Lemma~\ref{even-minus-lemma} we land up in the last case where both $p_{i-1}$ and $p_{i+2k+2}$ are above $l_{i+k}$, then we get a fourth compatible quadrilateral $p_3p_4p_5p_{10}$, otherwise we get $p_3p_4p_5p_7$. Notice that Assertion 1 holds in this case. Later, we will describe how to handle the remaining two cases.

 \item $m(\delta)=4$: Since $\delta$ starts with $+-$ or $+++$, and ends with $-+$ or $+++$, we only have sequences $+---++-+$, $+----+++$, $+--++--+$, $+--+-+-+$, $+-+--+-+$, and their reverses. For the sequence $+---++-+$ we get one quadrilateral for the plus signs and three for the minus signs. The sequence $+----+++$ can be handled by splitting as $+----+|++$; we get three quadrilaterals for the first subsequence and one for the second subsequence. Later, we will describe how to handle the remaining three cases.
 \item $m(\delta)=5$: In this case $\delta$ starts with $+-$ and ends with $-+$. The third plus sign partitions the minus signs into two subsequences of length one and four, or two and three. Thus, only the sequences $+-+----+$, $+--+---+$, and their reverses are valid. Later, we will show how to handle these two cases. 
 \item $m(\delta)=6$: In this case $\delta=+------+$, and by Lemma~\ref{even-minus-lemma} we get four quadrilaterals.
\end{itemize}

Now, we show how to get four compatible quadrilaterals for each of the sequences $+-+--+++$, $+--+-+++$, $+--++--+$, $+--+-+-+$, $+-+--+-+$, $+-+----+$, and $+--+---+$.
\begin{itemize}
  \item $+-+--+++$: If $p_8$ is below $l_{4,7}$ as in Figure~\ref{11-pmpmmppp-fig}(a), then we get four compatible quadrilaterals $p_0p_2p_3p_4$, $p_4p_5p_6p_7$, $p_6p_9p_8p_7$, and $p_0p_4p_7p_8$. Assume $p_8$ is above $l_{4,7}$ as in Figure~\ref{11-pmpmmppp-fig}(b). Since $p_5$ is above $l_{4,7}$, $p_8$ is also above $l_{5,7}$. Since $p_8$ and $p_9$ have $+$ signatures, the points $p_5,p_7,p_8,p_9,p_{10}$ are in convex position. The point $p_6$ is either in the convex hull of these five points or above $l_{5,10}$. In either case, by Lemma~\ref{six-lemma} we get two compatible quadrilaterals above $l_{5,7}$. We get two other compatible quadrilaterals $p_0p_2p_3p_4$ and $p_0p_4p_5p_7$ below $l_{5,7}$. 

\begin{figure}[htb]
  \centering
\setlength{\tabcolsep}{0in}
  $\begin{tabular}{cc}
 \multicolumn{1}{m{.5\columnwidth}}{\centering\includegraphics[width=.35\columnwidth]{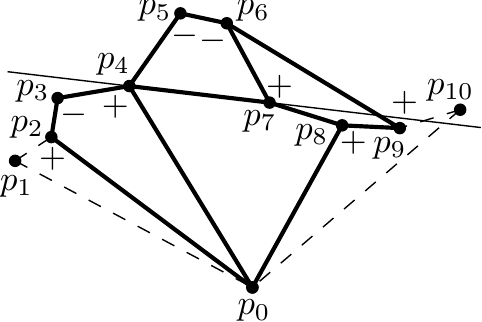}}
&\multicolumn{1}{m{.5\columnwidth}}{\centering\includegraphics[width=.35\columnwidth]{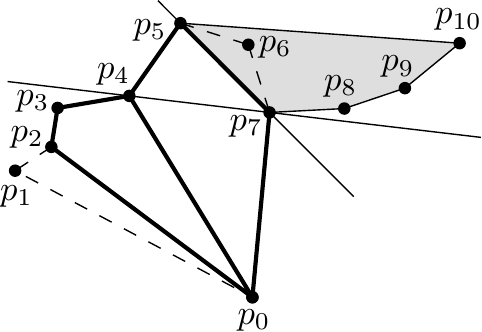}}
\\
(a) & (b)
\end{tabular}$
  \caption{The sequence $+-+--+++$. (a) $p_8$ is below $l_{4,7}$. (b) $p_8$ is above $l_{4,7}$.}
\label{11-pmpmmppp-fig}
\end{figure}

  \item $+--+-+++$: First we look at $p_1$. If $p_1$ is below $l_{2,5}$, then $p_0p_1p_2p_5$, $p_2p_3p_4p_5$, $p_0p_5p_6p_7$, and $p_7p_{10}p_9p_8$ are compatible. If $p_1$ is above $l_{2,5}$ and below $l_{3,4}$, then $p_1p_3p_4p_2$, $p_0p_2p_4p_5$, $p_0p_5p_6p_7$, and $p_7p_{10}p_9p_8$ are compatible. Assume $p_1$ is above $l_{3,4}$. Now we look at $p_6$. If $p_6$ is above $l_{3,4}$, then $p_1p_6p_4p_3$, $p_2p_3p_4p_5$, $p_0p_5p_6p_7$, and $p_7p_{10}p_9p_8$ are compatible. If $p_6$ is below $l_{3,4}$ and above $l_{2,5}$, then $p_6p_5p_3p_4$, $p_0p_2p_3p_5$, $p_0p_5p_6p_7$, and $p_7p_{10}p_9p_8$ are compatible. Assume $p_6$ is below $l_{2,5}$. Since $p_6$ has $-$ signature, $p_7$ is also below $l_{2,5}$. Now we look at $p_8$. If $p_8$ is below $l_{5,7}$ as in Figure~\ref{11-pmm-fig}(a), then $p_2p_3p_4p_5$, $p_2p_5p_6p_7$, $p_0p_2p_7p_8$, and $p_7p_{10}p_9p_8$ are compatible. Assume $p_8$ is above $l_{5,7}$. Since $p_8$ and $p_9$ have $+$ signatures, the points $p_5,p_7,p_8,p_9,p_{10}$ are in convex position. The point $p_6$ is either in the convex hull of these five points or above $l_{5,10}$. In either case, by Lemma~\ref{six-lemma} we get two quadrilaterals above $l_{5,7}$; These two quadrilaterals are compatible with $p_0p_2p_5p_7$ and $p_2p_3p_4p_5$.

   \item $+--++--+$: Notice that a $-$ sign introduces one quadrilateral, and a subsequence $++$ also introduces a quadrilateral. As in the previous case if $p_1$ is below $l_{2,5}$, or above $l_{2,5}$ and below $l_{3,4}$, then we get one extra compatible quadrilateral. Assume that $p_1$ is above $l_{3,4}$. Similarly we can assume that $p_{10}$ is above $l_{7,8}$. If $p_7$ (or $p_8$) is above $l_{3,4}$, then we get one extra compatible quadrilateral $p_1p_7p_4p_3$ (or $p_1p_8p_4p_3$). Similarly, if $p_4$ or $p_3$ is above $l_{7,8}$, then we can get one extra compatible quadrilateral. Assume $p_7$ and $p_8$ are below $l_{3,4}$, and $p_3$ and $p_4$ are below $l_{7,8}$ as in Figure~\ref{11-pmm-fig}(b). If $p_2$ (or $p_9$) is below $l_{5,6}$, then we get an extra compatible quadrilateral $p_0p_2p_5p_6$ (or $p_0p_5p_6p_9$). Assume both $p_2$ and $p_9$ are above $l_{5,6}$, then $p_3$ and $p_8$ are also above $l_{5,6}$ as in Figure~\ref{11-pmm-fig}(b). In this case $p_3,p_4,p_5,p_6,p_7,p_8$ are in convex position and by Lemma~\ref{six-lemma} we get two quadrilaterals; these two quadrilaterals are compatible with $p_0p_2p_3p_5$ and $p_0p_6p_8p_9$. 

\begin{figure}[htb]
	\centering
	\setlength{\tabcolsep}{0in}
	$\begin{tabular}{cc}
	\multicolumn{1}{m{.5\columnwidth}}{\centering\includegraphics[width=.35\columnwidth]{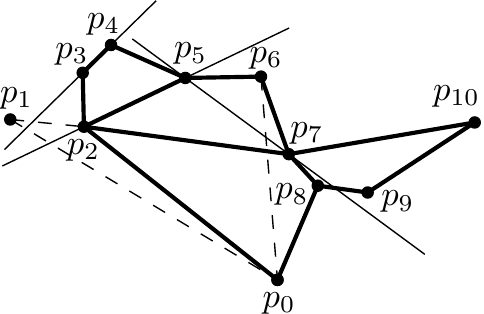}}
	&\multicolumn{1}{m{.5\columnwidth}}{\centering\includegraphics[width=.35\columnwidth]{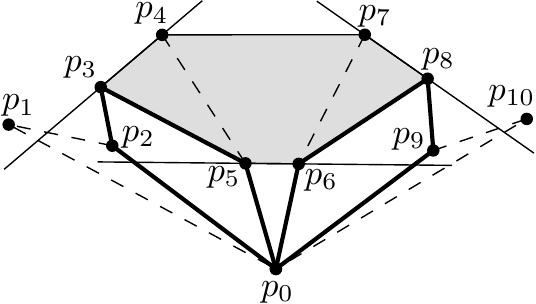}}
	\\
	(a) & (b)
	\end{tabular}$
	\caption{(a) The sequence $+--+-+++$ where $p_1$ is above $l_{3,4}$, $p_6$ is below $l_{2,5}$, and $p_8$ is below $l_{5,7}$. (b) The sequence $+--++--+$, where where $p_1$ is above $l_{3,4}$, $p_{10}$ is above $l_{7,8}$, $p_3$ and $p_4$ are below $l_{7,8}$, $p_7$ and $p_8$ are below $l_{3,4}$, and both $p_3$ and $p_8$ are above $l_{5,6}$.}
	\label{11-pmm-fig}
\end{figure}

\item $+--+-+-+$: Each of $p_6$ and $p_8$, which have $-$ signatures, introduces a quadrilateral. As in the previous cases, if $p_1$ is below $l_{2,5}$, or above $l_{2,5}$ and below $l_{3,4}$, then we get two other quadrilaterals. Assume that $p_1$ is above $l_{3,4}$. If $p_6$ is above $l_{3,4}$, or below $l_{3,4}$ and above $l_{2,5}$, then we get two other quadrilaterals. Assume that $p_6$ is below $l_{2,5}$; consequently $p_7$ is below $l_{2,5}$. See Figure~\ref{11-pmm-fig2}(a). Now we look at $p_8$. If $p_8$ is above $l_{5,6}$, then we get two other quadrilaterals $p_2p_3p_4p_5$ and $p_4p_8p_6p_5$. If $p_8$ is below $l_{5,6}$ and above $l_{5,7}$, then $p_0p_2p_5p_7$, $p_2p_3p_4p_5$, $p_0p_7p_8p_9$, and $p_5p_6p_8p_7$ are compatible. Assume $p_8$ is below $l_{5,7}$; consequently $p_9$ is below $l_{5,7}$ as in Figure~\ref{11-pmm-fig2}(a). Now we look at $p_{10}$. If $p_{10}$ is below $l_{7,9}$, then $p_0p_2p_9p_{10}$, $p_2p_3p_4p_5$, $p_2p_5p_6p_7$, and $p_2p_7p_8p_9$ are compatible. If $p_{10}$ is above $l_{7,9}$ and below $l_{7,8}$ as in Figure~\ref{11-pmm-fig2}(a), then $p_7p_8p_{10}p_9$, $p_2p_3p_4p_5$, $p_2p_5p_6p_7$, and $p_0p_2p_7p_9$ are compatible. Otherwise, $p_{10}$ is above $l_{7,8}$, and thus, $p_6p_{10}p_8p_7$, $p_2p_3p_4p_5$, $p_0p_5p_6p_7$, and $p_0p_7p_8p_9$ are compatible.

\begin{figure}[htb]
  \centering
\setlength{\tabcolsep}{0in}
  $\begin{tabular}{cc}
 \multicolumn{1}{m{.5\columnwidth}}{\centering\includegraphics[width=.35\columnwidth]{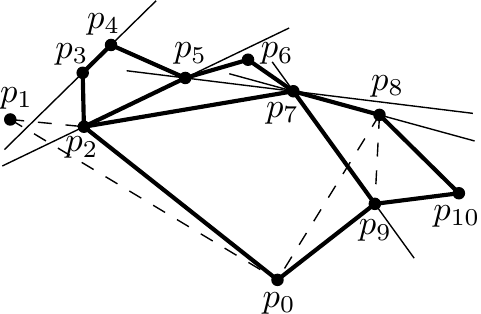}}
&\multicolumn{1}{m{.5\columnwidth}}{\centering\includegraphics[width=.35\columnwidth]{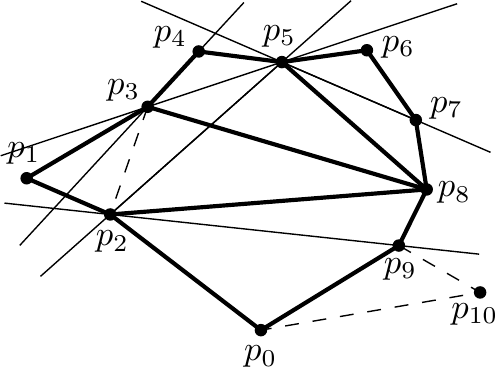}}
\\
(a) & (b)
\end{tabular}$
  \caption{The sequence $+--+-+-+$ where $p_1$ is above $l_{3,4}$, $p_6$ is below $l_{2,5}$, $p_8$ is below $l_{5,7}$, and $p_{10}$ is below $l_{7,8}$ and above $l_{7,9}$. (b) The sequence $+--+---+$ where $p_1$ is above $l_{3,4}$ and $l_{2,9}$ and below $l_{3,5}$, $p_6$ is below $l_{3,5}$ and $l_{3,5}$, and $p_4$ is below $l_{5,7}$.}
\label{11-pmm-fig2}
\end{figure}

  \item $+-+--+-+$: Each of $p_3$ and $p_8$, which have $-$ signatures, introduces a quadrilateral. If none of $p_3$ and $p_8$ is below $l_{4,7}$, then as in previous cases we get two other compatible quadrilaterals. Thus, assume without loss of generality that $p_8$ is below $l_{4,7}$; consequently $p_9$ is below $l_{4,7}$. If $p_{10}$ is below $l_{7,9}$, then $p_0p_2p_3p_4$, $p_4p_5p_6p_7$, $p_4p_7p_8p_9$, and $p_0p_4p_9p_{10}$ are compatible. If $p_{10}$ is above $l_{7,9}$ and below $l_{8,8}$, then $p_7p_8p_{10}p_9$, $p_4p_5p_6p_7$, $p_0p_4p_7p_9$, and $p_0p_2p_3p_4$ are compatible. Otherwise, $p_{10}$ is above $l_{7,8}$, and thus, $p_0p_2p_3p_4$, $p_4p_5p_6p_7$, $p_0p_7p_8p_9$, and  $p_6p_{10}p_8p_7$ are compatible.

  \item $+-+----+$: The point $p_3$ introduces a quadrilateral, and the four consecutive minus signs introduce two quadrilaterals. If $p_{10}$ is below any of the lines $l_{4,9}$, $l_{5,8}$, and $l_{6,7}$, then as in the proof of Lemma~\ref{even-minus-lemma} we get one extra compatible quadrilateral. Assume that $p_{10}$ is above all these lines. Now, if $p_3$ is above any of $l_{4,9}$, $l_{5,8}$, and $l_{6,7}$, again as in the proof of Lemma~\ref{even-minus-lemma} we get one extra compatible quadrilateral. Assume that $p_3$ is below these lines and specifically below $l_{4,9}$; consequently $p_2$ is below $l_{4,9}$. Since $p_2$ has $+$ signature, $p_1$ is above $l_{2,3}$. Now, if $p_1$ is above $l_{3,4}$, then we get an extra compatible quadrilateral $p_1p_5p_4p_3$. Otherwise, $p_1$ is below $l_{3,4}$ and above $l_{2,4}$, and thus, $p_1p_3p_4p_2$, $p_0p_2p_4p_9$, $p_4p_5p_8p_9$, and $p_5p_6p_7p_8$ are compatible. 
  \item $+--+---+$: Notice that we can get three compatible quadrilaterals below the chain $p_2,p_4,\dots, p_9$ (and $p_2,p_3,p_5,\dots, p_9$). First we look at $p_1$. If $p_1$ is below $l_{2,5}$, the $p_2p_3p_4p_5$, $p_0p_1p_2p_5$, $p_0p_5p_6p_7$, and $p_0p_7p_8p_9$ are compatible. Assume that $p_1$ is above $l_{2,5}$. If $p_1$ is below $l_{3,4}$, then we get one extra compatible quadrilateral $p_1p_3p_4p_2$. Assume that $p_1$ is above $l_{3,4}$. If $p_1$ is above $l_{4,5}$, then $p_1p_6p_5p_4$ is an extra compatible quadrilateral. If $p_1$ is below $l_{4,5}$ and above $l_{3,5}$, then $p_1p_4p_5p_3$ is an extra compatible quadrilateral. Assume that $p_1$ is below $l_{3,5}$. If $p_1$ is below $l_{2,9}$, then $p_0p_1p_2p_9$, $p_2p_3p_4p_5$, $p_5p_6p_7p_8$, and $p_2p_5p_8p_9$ are compatible. Assume that $p_1$ is above $l_{2,9}$. See Figure~\ref{11-pmm-fig2}(b). Now we look at $p_6$. If $p_6$ is above $l_{3,4}$, then we get an extra compatible quadrilateral $p_1p_6p_4p_3$. If $p_6$ is below $l_{3,4}$ and above $l_{3,5}$, then we get an extra compatible quadrilateral $p_3p_4p_6p_5$. Assume that $p_6$ is below $l_{3,5}$; consequently $p_7$, $p_8$, and $p_9$ are below $l_{3,5}$. If $p_4$ is above $l_{7,5}$, then $p_4p_6p_7p_5$, $p_2p_3p_4p_5$, $p_5p_7p_8p_9$, and $p_0p_2p_5p_9$ are compatible. Assume that $p_4$, and consequently $p_3$ are below $l_{5,7}$ as in Figure~\ref{11-pmm-fig2}(b). In this case $\cone{p_1}{p_2}{p_5}$ contains $p_6$ $p_7$, $p_8$, and $p_9$. Thus, $p_5p_6p_7p_8$, $p_3p_4p_5p_8$, $p_1p_3p_8p_2$, and $p_0p_2p_8p_9$ are compatible.

\end{itemize}

\section{Compatible 4-holes in $n$-sets}
\label{n-section}
In this section we prove our main claim for large point sets, that is, every $n$-set contains at least $\lfloor 5n/11\rfloor -1$ compatible 4-holes. As in Section~\ref{preliminary-section}, by combining Theorems~\ref{9-thr} and \ref{11-thr} with the idea of partitioning the points into some cones with respect to their radial ordering about a point $p_0$, we can improve the lower bound on the number of compatible 4-holes in an $n$-set to $3\cdot\lfloor(n-2)/7\rfloor$ and $4\cdot\lfloor(n-2)/9\rfloor$, respectively. In the rest of this section, we first prove a lemma, that can be used to improve these bounds further. We denote by $ab$ the straight-line through two points $a$ and $b$. We say that a 4-hole $Q$ is {\em compatible with} a point set $A$ if the interior of $Q$ is disjoint from the interior of the convex hull of $A$.

\begin{lemma}
\label{extraQ-lemma}
For every $(r{+}s)$-set, with $r,s\geqslant 4$, we can divide the plane into two internally disjoint convex regions such that one region contains a set $A$ of at least $s$ points, the other region contains a set $B$ of at least $r$ points, and there exists a 4-hole that is compatible with $A$ and $B$.
\end{lemma}

Before proving this lemma, we note that a similar lemma has been proved by Hosono and Urabe~(Lemma 3 in~\cite{Hosono2001}) for disjoint 4-holes, where they obtain a set $A'$ of $s{-}2$ points, a set $B'$ of $r{-}2$ points, and a 4-hole $Q$ that is disjoint from $A'$ and $B'$. However, their lemma does not imply our Lemma~\ref{extraQ-lemma}, because it might not be possible to add two points of $Q$ to $A'$ to obtain a set $A$ of $s$ points such that $Q$ is compatible with $A$.

In the following proof, if there exist two internally disjoint convex regions such that one of them contains a set $A$ of $s$ points, the other contains a set $B$ of $r$ points, and there exists a 4-hole that is compatible with $A$ and $B$, then we say that $A$ and $B$ are {\em good}.

\begin{proof}[Proof of Lemma~\ref{extraQ-lemma}]
Consider an $(r{+}s)$-set. In this proof a ``point'' refers to a point from this set. Also when we say that a convex shape is ``empty'' we mean that its interior does not contain any point from this set.  

\begin{figure}[htb]
  \centering
\setlength{\tabcolsep}{0in}
  $\begin{tabular}{ccc}
 \multicolumn{1}{m{.33\columnwidth}}{\centering\includegraphics[width=.32\columnwidth]{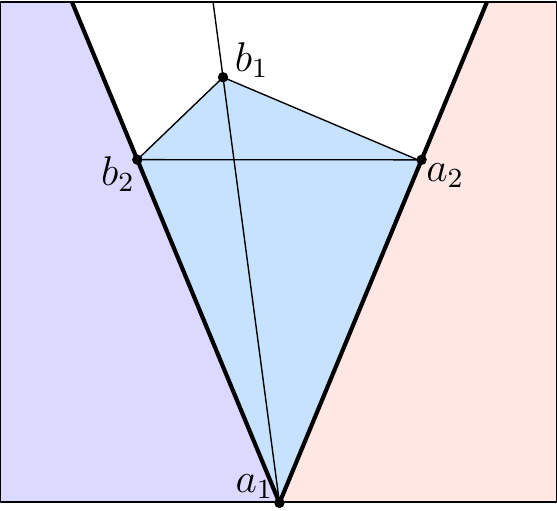}}
&\multicolumn{1}{m{.33\columnwidth}}{\centering\includegraphics[width=.32\columnwidth]{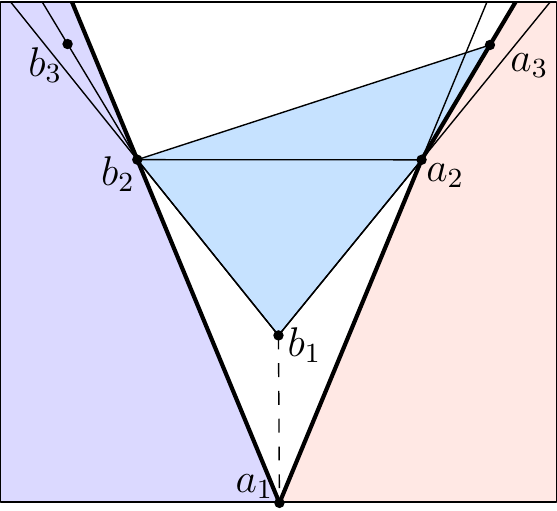}}
&\multicolumn{1}{m{.33\columnwidth}}{\centering\includegraphics[width=.32\columnwidth]{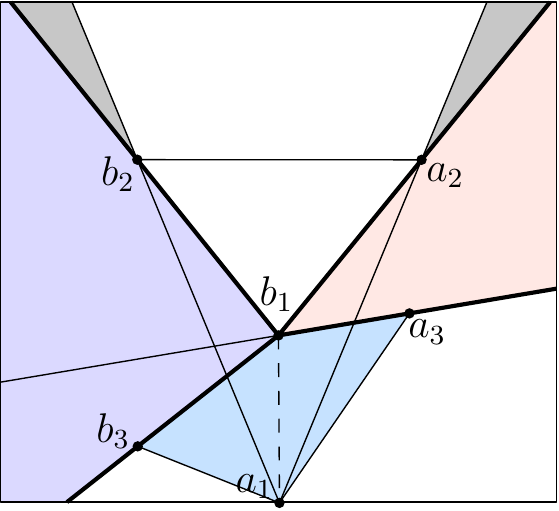}}
\\
(a)&(b)&(c)\\
 \multicolumn{1}{m{.33\columnwidth}}{\centering\includegraphics[width=.32\columnwidth]{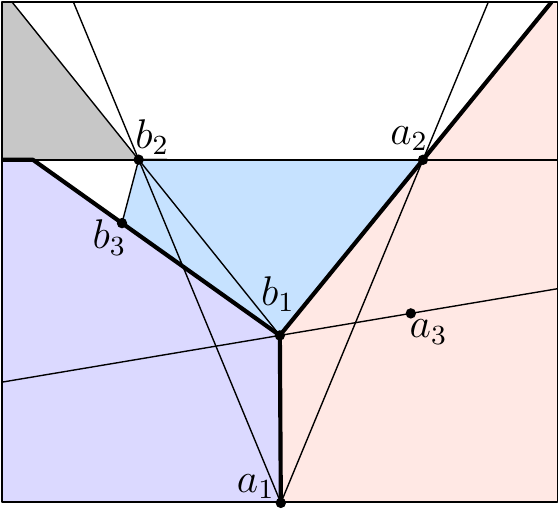}}
&\multicolumn{1}{m{.33\columnwidth}}{\centering\includegraphics[width=.32\columnwidth]{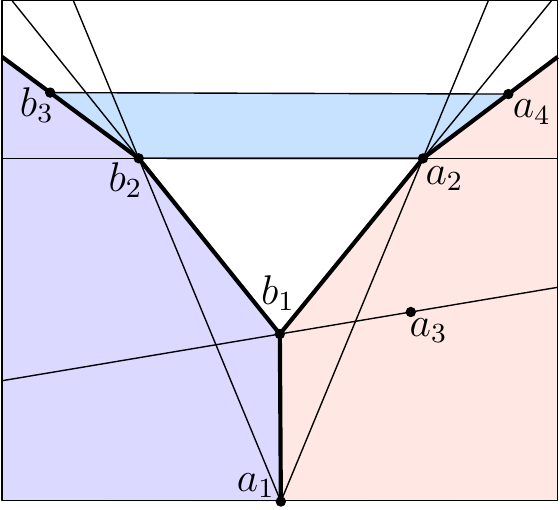}}
&\multicolumn{1}{m{.33\columnwidth}}{\centering\includegraphics[width=.32\columnwidth]{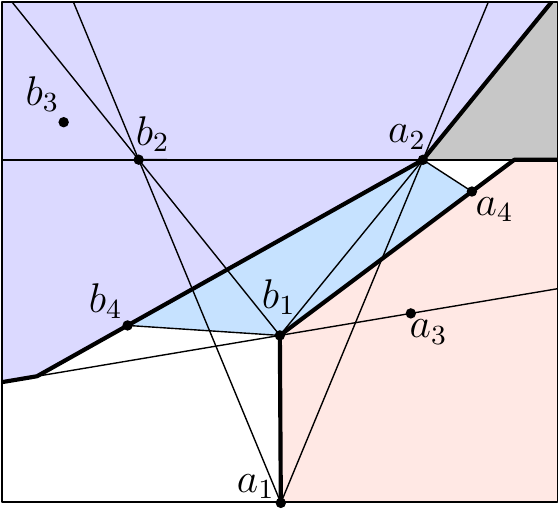}}
\\
(d)&(e)&(f)
\end{tabular}$
  \caption{Illustration of Lemma~\ref{extraQ-lemma}. The convex regions with $r$ and $s$ points are shown in light purple and light orange colors, respectively. The compatible 4-holes with these regions are in blue color. The grey regions are empty.}
\label{extraQ-fig}
\end{figure}

Let $a_1$ be a point on the convex hull of this set, and without loss of generality assume that $a_1$ is the lowest point. Let $a_2$ be the point such that $s{-}2$ points are to the right side of the line $a_1a_2$. Let $A$ be the set of points that are on or to the right side of $a_1a_2$, and let $B$ be the set of other points. Notice that $A$ contains $s$ points and $B$ contains $r$ points. Let $b_1$ be the point of $B$ such that the interior of $\cone{a_1}{a_2}{b_1}$ does not contain any point. Let $b_2$ be the point of $B$ such that the interior of $\cone{a_1}{a_2}{b_2}$ contains only $b_1$. See Figure~\ref{extraQ-fig}(a).

If $b_1$ is not in the interior of the triangle $\bigtriangleup a_1a_2b_2$, then $a_1a_2b_1b_2$ is a 4-hole that is compatible with $A$ and $(B\setminus\{b_1\})\cup\{a_1\}$. As shown in Figure~\ref{extraQ-fig}(a), the interiors of the convex hulls of these two sets are disjoint, and thus, these two sets are good. Assume that $b_1$ is in the interior of $\bigtriangleup a_1a_2b_2$. We consider two cases depending on whether or not $\cone{b_1}{b_2}{a_2}$ is empty.  
\begin{itemize}
 \item $\cone{b_1}{b_2}{a_2}$ is not empty. If $\cone{b_1}{b_2}{a_2}$ contains a point of $A$, then let $a_3$ be such a point that is the neighbor of $a_2$ on $\CH{A}$; see Figure~\ref{extraQ-fig}(b). Then $b_1b_2a_3a_2$ is a 4-hole, and $A$ and $(B\setminus\{b_1\})\cup\{a_1\}$ are good. If $\cone{b_1}{b_2}{a_2}$ contains a point of $B$, then let $b_3$ be such a point that is the neighbor of $b_2$ on $\CH{B}$. Then $b_1b_2b_3a_2$ is a 4-hole, and $A$ and $(B\setminus\{b_1\})\cup\{a_1\}$ are good. 
\item $\cone{b_1}{b_2}{a_2}$ is empty. Let $a_3$ be the attack point of $\rotatingl{b_1}{a_1}{a_2}$; recall that this is the first point that $\rotatingl{b_1}{a_1}{a_2}$ meets. If the attack point of $\rotatingl{b_1}{a_1}{b_2}$ is below $b_1a_3$, then let $b_3$ be that point; Figure~\ref{extraQ-fig}(c). In this case $b_1a_3a_1b_3$ is a 4-hole, and $(A\setminus\{a_1\})\cup\{b_1\}$ and $B$ are good. Assume that the attack point of $\rotatingl{b_1}{a_1}{b_2}$ is above $b_1a_3$. We consider the following two cases depending on whether or not there is a point of $B$ above the line $a_2b_2$.
\begin{itemize}
  \item No point of $B$ is above $a_2b_2$. Let $b_3$ be the attack point of $\rotatingl{b_1}{b_2}{a_1}$ as in Figure~\ref{extraQ-fig}(d). Then $b_1b_3b_2a_2$ is a 4-hole, and $A\cup\{b_1\}$ and $(B\setminus\{b_2\})\cup\{a_1\}$ are good.

  \item Some point of $B$ is above $a_2b_2$. Let $b_3$ be such a point that is the neighbor of $b_2$ on $\CH{B}$ as in Figure~\ref{extraQ-fig}(e). If some point of $A$ is above $a_2b_2$, then let $a_4$ be such a point that is the neighbor of $a_2$ on $\CH{A}$; see Figure~\ref{extraQ-fig}(e). Then $a_2b_2b_3a_4$ is a 4-hole, and $A\cup\{b_1\}$ and $B\cup\{a_1\}$ are good. Assume that no point of $A$ is above $a_2b_2$. Let $a_4$ be the attack point of $\rotatingl{b_1}{a_2}{a_3}$ and $b_4$ be the attack point of $\rotatingl{a_2}{b_1}{b_2}$ as in Figure~\ref{extraQ-fig}(f). Notice that it might be the case that $b_4=b_2$. In either case, $b_1b_4a_2a_4$ is a 4-hole, and $(A\setminus\{a_2\})\cup\{b_1\}$ and $(B\setminus\{b_1\})\cup\{a_2\}$ are good.
\end{itemize}
\end{itemize}
\end{proof}

\begin{theorem}
Every $n$-set contains at least $\lfloor 5n/11\rfloor - 1$  compatible 4-holes. 
\end{theorem}
\begin{proof}
Let $P$ be an $n$-set. Our proof is by induction on the number of points in $P$. The base cases happen when $|P|\leqslant 14$. If $|P|\leqslant 13$, then our claim follows from one of Theorems~\ref{Klein-thr}, \ref{Sakai-thr}, \ref{9-thr}, or \ref{11-thr}. If $|P|=14$, then by applying Lemma~\ref{extraQ-lemma} on $P$ with $r=s=7$ we get a 4-hole together with two sets $A$ and $B$ each containing at least 7 points. By Theorem~\ref{Sakai-thr} we get two 4-holes in each of $A$ and $B$. Thus, we get five compatible 4-holes in total. This finishes our proof for the base cases.

Assume that $|P|\geqslant 15$. By applying Lemma~\ref{extraQ-lemma} on $P$ with $r=n{-}11$  and $s=11$ (notice that $r$ is at least four as required by this lemma) we get a 4-hole together with two sets $A$ and $B$ such that the interiors of their convex hulls are disjoint, $A$ contains at least 11 points, and $B$ contains at least $n{-}11$ points.  
By Theorem~\ref{11-thr} we get four compatible 4-holes in $\CH{A}$. By induction, we get $\lfloor 5(n-11)/11\rfloor - 1$ compatible 4-holes in $\CH{B}$. Therefore, in total, we get $$1+4+\left(\left\lfloor \frac{5(n-11)}{11}\right\rfloor - 1\right)=\left\lfloor\frac{5n}{11}\right\rfloor - 1$$ compatible 4-holes in $P$.  
\end{proof}

An $O(n\log^2 n)$-time algorithm for computing this many 4-holes follows from the proofs, by using a dynamic convex hull data structure for computing the sets $A$ and $B$ in Lemma~\ref{extraQ-lemma}. 
\removed{
\section{Conclusions}
\label{conclusions-section}
In this paper we studied the problem of finding the maximum number compatible 4-holes in any $n$-set. The best possible bound is at most $\lceil n/2\rceil{-}2$ as for every $n$ there exists an $n$-set with at most this many compatible 4-holes. We proved tight bounds for 9-sets and 11-sets. Specifically, we showed that every 11-set contains at least four compatible 4-holes. By combining this with Lemma~\ref{extraQ-lemma} we proved that every $n$-set contains at least $\lfloor 5n/11\rfloor {-} 1$ compatible 4-holes. A natural open problem is to improve this bound. By doing a more extensive case analysis, it may be possible to show that every 13-set contains at least five compatible 4-holes, and thus, by using Lemma~\ref{extraQ-lemma}, this would improve the lower bound to $\lfloor 6n/13\rfloor {-} 1$. However, the main open problem is to prove or disprove our following conjecture:
 
\begin{conjecture}
Every $n$-set contains at least $n/2 -O(1)$ compatible 4-holes.
\end{conjecture}
}
\bibliographystyle{abbrv}
\bibliography{Convex-Quadrilaterals.bib}

\begin{thebibliography}{10}

\bibitem{Urrutia2017}
Personal communication with {J.} {U}rrutia.

\bibitem{Aichholzer2006}
O.~Aichholzer, F.~Aurenhammer, and H.~Krasser.
\newblock On the crossing number of complete graphs.
\newblock {\em Computing}, 76(1):165--176, 2006.

\bibitem{Aichholzer2007}
O.~Aichholzer, C.~Huemer, S.~Kappes, B.~Speckmann, and C.~D. T{\'{o}}th.
\newblock Decompositions, partitions, and coverings with convex polygons and
  pseudo-triangles.
\newblock {\em Graphs and Combinatorics}, 23(5):481--507, 2007.

\bibitem{Aichholzer2001}
O.~Aichholzer and H.~Krasser.
\newblock The point set order type data base: {A} collection of applications
  and results.
\newblock In {\em Proceedings of the 13th Canadian Conference on Computational
  Geometry}, pages 17--20, 2001.

\bibitem{Aichholzer2015}
O.~Aichholzer, R.~F. Monroy, H.~Gonz{\'{a}}lez{-}Aguilar, T.~Hackl, M.~A.
  Heredia, C.~Huemer, J.~Urrutia, P.~Valtr, and B.~Vogtenhuber.
\newblock On $k$-gons and $k$-holes in point sets.
\newblock {\em Comput. Geom.}, 48(7):528--537, 2015.

\bibitem{Bhattacharya2011}
B.~B. Bhattacharya and S.~Das.
\newblock On the minimum size of a point set containing a 5-hole and a disjoint
  4-hole.
\newblock {\em Studia Scientiarum Mathematicarum Hungarica}, 48(4):445--457,
  2011.

\bibitem{Bhattacharya2013}
B.~B. Bhattacharya and S.~Das.
\newblock Disjoint empty convex pentagons in planar point sets.
\newblock {\em Periodica Mathematica Hungarica}, 66(1):73--86, 2013.

\bibitem{Biniaz2017}
A.~Biniaz, A.~Maheshwari, and M.~Smid.
\newblock Compatible 4-holes in point sets.
\newblock {\em CoRR}, abs/1706.08105, 2017.

\bibitem{Bose2002}
P.~Bose, S.~Ramaswami, G.~T. Toussaint, and A.~Turki.
\newblock Experimental results on quadrangulations of sets of fixed points.
\newblock {\em Computer Aided Geometric Design}, 19(7):533--552, 2002.

\bibitem{Bose1997}
P.~Bose and G.~T. Toussaint.
\newblock Characterizing and efficiently computing quadrangulations of planar
  point sets.
\newblock {\em Computer Aided Geometric Design}, 14(8):763--785, 1997.

\bibitem{Brodsky2003}
A.~Brodsky, S.~Durocher, and E.~Gethner.
\newblock Toward the rectilinear crossing number of ${K}_n$: new drawings,
  upper bounds, and asymptotics.
\newblock {\em Discrete Mathematics}, 262(1-3):59--77, 2003.

\bibitem{Cano2015}
J.~Cano, A.~G. Olaverri, F.~Hurtado, T.~Sakai, J.~Tejel, and J.~Urrutia.
\newblock Blocking the $k$-holes of point sets in the plane.
\newblock {\em Graphs and Combinatorics}, 31(5):1271--1287, 2015.

\bibitem{Devillers2003}
O.~Devillers, F.~Hurtado, G.~K{\'{a}}rolyi, and C.~Seara.
\newblock Chromatic variants of the {E}rd{\H{o}}s-{S}zekeres theorem on points
  in convex position.
\newblock {\em Comput. Geom.}, 26(3):193--208, 2003.

\bibitem{Erdos1978}
P.~{E}rd{\H{o}}s.
\newblock Some more problems on elementary geometry.
\newblock {\em Austral. Math. Soc. Gaz.}, 5:52--54, 1978.

\bibitem{Erdos1935}
P.~{E}rd{\H{o}}s and G.~{S}zekeres.
\newblock A combinatorial problem in geometry.
\newblock {\em Compositio Mathematica}, 2:463--470, 1935.

\bibitem{Gerken2008}
T.~Gerken.
\newblock Empty convex hexagons in planar point sets.
\newblock {\em Discrete {\&} Computational Geometry}, 39(1):239--272, 2008.

\bibitem{Harborth1978}
H.~Harborth.
\newblock Konvexe {F\"{u}nfecke} in ebenen {Punktmengen}.
\newblock {\em Elemente der Mathematik}, 33:116--118, 1978.

\bibitem{Horton1983}
J.~D. Horton.
\newblock Sets with no empty convex 7-gons.
\newblock {\em Canad. Math. Bull.}, 26(4):482--484, 1983.

\bibitem{Hosono2001}
K.~Hosono and M.~Urabe.
\newblock On the number of disjoint convex quadrilaterals for a planar point
  set.
\newblock {\em Comput. Geom.}, 20(3):97--104, 2001.

\bibitem{Lomeli2007}
M.~Lomeli-Haro, T.~Sakai, and J.~Urrutia.
\newblock Convex quadrilaterals of point sets with disjoint interiors.
\newblock In {\em Abstracts of Kyoto International Conference on Computational
  Geometry and Graph Theory (KyotoCGGT2007)}.

\bibitem{Lovasz2004}
L.~Lov\'{a}sz, K.~Vesztergombi, U.~Wagner, and E.~Welzl.
\newblock Convex quadrilaterals and $k$-sets.
\newblock In J.~Pach, editor, {\em Towards a Theory of Geometric Graphs},
  volume 342 of {\em Contemporary Mathematics}, pages 139--148. 2004.

\bibitem{Nicolas2007}
C.~M. Nicol\'{a}s.
\newblock The empty hexagon theorem.
\newblock {\em Discrete {\&} Computational Geometry}, 38(2):389--397, 2007.

\bibitem{Ramaswami1998}
S.~Ramaswami, P.~A. Ramos, and G.~T. Toussaint.
\newblock Converting triangulations to quadrangulations.
\newblock {\em Comput. Geom.}, 9(4):257--276, 1998.

\bibitem{Sakai2007}
T.~Sakai and J.~Urrutia.
\newblock Covering the convex quadrilaterals of point sets.
\newblock {\em Graphs and Combinatorics}, 23(Supplement-1):343--357, 2007.

\bibitem{Toussaint1995}
G.~T. Toussaint.
\newblock Quadrangulations of planar sets.
\newblock In {\em Proceedings of the 4th International Workshop on Algorithms
  and Data Structures $(${WADS}$)$}, pages 218--227, 1995.

\bibitem{Valtr2008}
P.~Valtr.
\newblock On empty hexagons.
\newblock In J.~E. Goodman, J.~Pach, and R.~Pollack, editors, {\em Surveys on
  Discrete and Computational Geometry: Twenty Years Later}, pages 433--441.
  2008.

\bibitem{Vogtenhuber2011}
B.~Vogtenhuber.
\newblock {\em Combinatorial aspects of colored point sets in the plane}.
\newblock PhD thesis, Graz University of Technology, November 2011.

\bibitem{Wagner2003}
U.~Wagner.
\newblock On the rectilinear crossing number of complete graphs.
\newblock In {\em Proceedings of the 14th Annual {ACM-SIAM} Symposium on
  Discrete Algorithms $(${SODA}$)$}, pages 583--588, 2003.

\bibitem{Wu2008}
L.~Wu and R.~Ding.
\newblock On the number of empty convex quadrilaterals of a finite set in the
  plane.
\newblock {\em Appl. Math. Lett.}, 21(9):966--973, 2008.

\end{thebibliography}
\end{document}